\documentclass[submission,copyright,creativecommons]{eptcs}
\usepackage[utf8x]{inputenc}
\providecommand{\event}{QPL 2020} 

\usepackage{graphicx}
\usepackage{subcaption}
\usepackage{dcolumn}
\usepackage{bm}
\usepackage{tikz-cd}
\usepackage{tikz}
\usepackage{bbm}
\usepackage{physics}
\usepackage{dsfont}
\usepackage{float}
\usepackage{amsthm}
\usepackage{circuitikz}
\usepackage{tikzit}
\usepackage{url}

\usepackage{amsmath}
\usepackage{amsfonts}
\usepackage{amssymb}
\usepackage{calrsfs} 
\DeclareMathAlphabet{\pazocal}{OMS}{zplm}{m}{n}

\DeclareFontFamily{U}{mathb}{\hyphenchar\font45}
\DeclareFontShape{U}{mathb}{m}{n}{
      <5> <6> <7> <8> <9> <10> gen * mathb
      <10.95> mathb10 <12> <14.4> <17.28> <20.74> <24.88> mathb12
      }{}
\DeclareSymbolFont{mathb}{U}{mathb}{m}{n}
\DeclareFontSubstitution{U}{mathb}{m}{n}
\DeclareFontFamily{U}{mathx}{\hyphenchar\font45}
\DeclareFontShape{U}{mathx}{m}{n}{
      <5> <6> <7> <8> <9> <10>
      <10.95> <12> <14.4> <17.28> <20.74> <24.88>
      mathx10
      }{}
\DeclareSymbolFont{mathx}{U}{mathx}{m}{n}
\DeclareFontSubstitution{U}{mathx}{m}{n}
\DeclareMathDelimiter{\thickvert}{0}{mathb}{"7E}{mathx}{"1F}

\tikzstyle{doubled}=[line width=1.5pt] 

\tikzstyle{dot}=[inner sep=0mm,minimum width=2mm,minimum height=2mm,draw,shape=circle]  
\tikzstyle{ddot}=[inner sep=0mm, doubled, minimum width=2.5mm,minimum height=2.5mm,draw,shape=circle]

\tikzstyle{pdot}=[inner sep=0mm, doubled, minimum width=2.5mm,minimum height=2.5mm,shape=circle]
\tikzstyle{phase dimensions}=[minimum size=6mm,font=\footnotesize,inner sep=0.2mm,outer sep=-2mm]

\tikzstyle{phase dot}=[pdot,phase dimensions]
\tikzstyle{wphase dot}=[dot, phase dimensions]

\tikzstyle{hadamard}=[fill=white,draw,inner sep=0.6mm,font=\footnotesize,minimum height=4mm,minimum width=4mm]

\tikzstyle{anti} = [shade, bottom color=black, top color = white, draw, minimum height = 4mm, minimum width = 4mm]

\tikzstyle{triang}=[regular polygon,regular polygon sides=3,draw,scale=0.75,inner sep=-0.5pt,minimum width=9mm,fill=white,regular polygon rotate=180]
\tikzstyle{triangdag}=[regular polygon,regular polygon sides=3,draw,scale=0.75,inner sep=-0.5pt,minimum width=9mm,fill=white]

\makeatletter
\newcommand{\boxshape}[3]{%
\pgfdeclareshape{#1}{
\inheritsavedanchors[from=rectangle] 
\inheritanchorborder[from=rectangle]
\inheritanchor[from=rectangle]{center}
\inheritanchor[from=rectangle]{north}
\inheritanchor[from=rectangle]{south}
\inheritanchor[from=rectangle]{west}
\inheritanchor[from=rectangle]{east}
\backgroundpath{
\southwest \pgf@xa=\pgf@x \pgf@ya=\pgf@y
\northeast \pgf@xb=\pgf@x \pgf@yb=\pgf@y

\@tempdima=#2
\@tempdimb=#3

\pgfpathmoveto{\pgfpoint{\pgf@xa - 5pt + \@tempdima}{\pgf@ya}}
\pgfpathlineto{\pgfpoint{\pgf@xa - 5pt - \@tempdima}{\pgf@yb}}
\pgfpathlineto{\pgfpoint{\pgf@xb + 5pt + \@tempdimb}{\pgf@yb}}
\pgfpathlineto{\pgfpoint{\pgf@xb + 5pt - \@tempdimb}{\pgf@ya}}
\pgfpathlineto{\pgfpoint{\pgf@xa - 5pt + \@tempdima}{\pgf@ya}}
\pgfpathclose
}
}}

\boxshape{NEbox}{0pt}{5pt}
\boxshape{SEbox}{0pt}{-5pt}
\boxshape{NWbox}{5pt}{0pt}
\boxshape{SWbox}{-5pt}{0pt}
\boxshape{EBox}{-3pt}{3pt}
\boxshape{WBox}{3pt}{-3pt}
\makeatother

\tikzstyle{map}=[draw,shape=SEbox,inner sep=2pt,minimum height=6mm,fill=white]
\tikzstyle{mapdag}=[draw,shape=NEbox,inner sep=2pt,minimum height=6mm,fill=white]
\tikzstyle{maptrans}=[draw,shape=NWbox,inner sep=2pt,minimum height=6mm,fill=white]
\tikzstyle{mapconj}=[draw,shape=SWbox,inner sep=2pt,minimum height=6mm,fill=white]

\tikzstyle{dmap}=[draw,doubled,shape=SEbox,inner sep=2pt,minimum height=6mm,fill=white]
\tikzstyle{dmapdag}=[draw,doubled,shape=NEbox,inner sep=2pt,minimum height=6mm,fill=white]
\tikzstyle{dmaptrans}=[draw,doubled,shape=NWbox,inner sep=2pt,minimum height=6mm,fill=white]
\tikzstyle{dmapconj}=[draw,doubled,shape=SWbox,inner sep=2pt,minimum height=6mm,fill=white]

\makeatletter
\pgfdeclareshape{cornerpoint}{
\inheritsavedanchors[from=rectangle] 
\inheritanchorborder[from=rectangle]
\inheritanchor[from=rectangle]{center}
\inheritanchor[from=rectangle]{north}
\inheritanchor[from=rectangle]{south}
\inheritanchor[from=rectangle]{west}
\inheritanchor[from=rectangle]{east}
\backgroundpath{
\southwest \pgf@xa=\pgf@x \pgf@ya=\pgf@y
\northeast \pgf@xb=\pgf@x \pgf@yb=\pgf@y

\pgfmathsetmacro{\pgf@shorten@left}{\pgfkeysvalueof{/tikz/shorten left}}
\pgfmathsetmacro{\pgf@shorten@right}{\pgfkeysvalueof{/tikz/shorten right}}

\pgfpathmoveto{\pgfpoint{0.5 * (\pgf@xa + \pgf@xb)}{\pgf@ya - 5pt}}
\pgfpathlineto{\pgfpoint{\pgf@xa - 8pt + \pgf@shorten@left}{\pgf@yb - 1.5 * \pgf@shorten@left}}
\pgfpathlineto{\pgfpoint{\pgf@xa - 8pt + \pgf@shorten@left}{\pgf@yb}}
\pgfpathlineto{\pgfpoint{\pgf@xb + 8pt - \pgf@shorten@right}{\pgf@yb}}
\pgfpathlineto{\pgfpoint{\pgf@xb + 8pt - \pgf@shorten@right}{\pgf@yb - 1.5 * \pgf@shorten@right}}
\pgfpathclose
}
}

\pgfdeclareshape{cornercopoint}{
\inheritsavedanchors[from=rectangle] 
\inheritanchorborder[from=rectangle]
\inheritanchor[from=rectangle]{center}
\inheritanchor[from=rectangle]{north}
\inheritanchor[from=rectangle]{south}
\inheritanchor[from=rectangle]{west}
\inheritanchor[from=rectangle]{east}
\backgroundpath{
\southwest \pgf@xa=\pgf@x \pgf@ya=\pgf@y
\northeast \pgf@xb=\pgf@x \pgf@yb=\pgf@y

\pgfmathsetmacro{\pgf@shorten@left}{\pgfkeysvalueof{/tikz/shorten left}}
\pgfmathsetmacro{\pgf@shorten@right}{\pgfkeysvalueof{/tikz/shorten right}}

\pgfpathmoveto{\pgfpoint{0.5 * (\pgf@xa + \pgf@xb)}{\pgf@yb + 5pt}}
\pgfpathlineto{\pgfpoint{\pgf@xa - 8pt + \pgf@shorten@left}{\pgf@ya + 1.5 * \pgf@shorten@left}}
\pgfpathlineto{\pgfpoint{\pgf@xa - 8pt + \pgf@shorten@left}{\pgf@ya}}
\pgfpathlineto{\pgfpoint{\pgf@xb + 8pt - \pgf@shorten@right}{\pgf@ya}}
\pgfpathlineto{\pgfpoint{\pgf@xb + 8pt - \pgf@shorten@right}{\pgf@ya + 1.5 * \pgf@shorten@right}}
\pgfpathclose
}
}

\makeatother

\pgfkeyssetvalue{/tikz/shorten left}{0pt}
\pgfkeyssetvalue{/tikz/shorten right}{0pt}

\tikzstyle{kpoint common}=[draw,fill=white,inner sep=1pt,minimum height=4mm]
\tikzstyle{kpoint}=[shape=cornerpoint,shorten left=5pt,kpoint common]
\tikzstyle{kpoint adjoint}=[shape=cornercopoint,shorten left=5pt,kpoint common]
\tikzstyle{kpoint conjugate}=[shape=cornerpoint,shorten right=5pt,kpoint common]
\tikzstyle{kpoint transpose}=[shape=cornercopoint,shorten right=5pt,kpoint common]

\tikzstyle{kpointdag}=[kpoint adjoint]
\tikzstyle{kpointadj}=[kpoint adjoint]
\tikzstyle{kpointconj}=[kpoint conjugate]
\tikzstyle{kpointtrans}=[kpoint transpose]

\tikzstyle{big kpoint}=[kpoint, minimum width=1.2 cm, minimum height=8mm, inner sep=4pt, text depth=3mm]

 \tikzstyle{upground}=[circuit ee IEC,thick,ground,rotate=90,scale=1.5]
 \tikzstyle{downground}=[circuit ee IEC,thick,ground,rotate=-90,scale=1.5]

\tikzstyle{discarding}=[fill=white, draw=black, shape=circle, style=upground]
\tikzstyle{smalldiscarding}=[fill=white, draw=black, style=upground, scale=0.5]
\tikzstyle{backdiscard}=[fill=white, draw=black, shape=circle, style=downground, scale=0.5]
\tikzstyle{smallbackdiscard}=[fill=white, draw=black, shape=circle, style=downground, scale=0.5]
\tikzstyle{state}=[fill=white, draw=black, style=triang, tikzit shape=rectangle]
\tikzstyle{kstate}=[fill=white, draw=black, style=kpoint, tikzit shape=rectangle]
\tikzstyle{kstateconj}=[fill=white, draw=black, style=kpoint conjugate, tikzit shape=rectangle]
\tikzstyle{kstateBIG}=[fill=white, draw=black, style=big kpoint, tikzit shape=rectangle]
\tikzstyle{effect}=[fill=white, draw=black, style=triangdag]
\tikzstyle{keffect}=[fill=white, draw=black, style=kpoint adjoint]
\tikzstyle{keffectconj}=[fill=white, draw=black, style=kpoint transpose]
\tikzstyle{morphdag}=[style=mapdag]
\tikzstyle{morph}=[style=hadamard]
\tikzstyle{WIDEmorph}=[style=hadamard, minimum width=14mm]
\tikzstyle{morphtrans}=[style=maptrans]
\tikzstyle{morphconj}=[style=mapconj]
\tikzstyle{CPMmorph}=[style=dmap]
\tikzstyle{CPMmorphconj}=[style=dmapconj]
\tikzstyle{CPMmorphdag}=[style=dmapdag]
\tikzstyle{CPMmorphtrans}=[style=dmaptrans]
\tikzstyle{CPMstate}=[fill=white, draw=black, style=triang, doubled]
\tikzstyle{CPMstateBIG}=[fill=white, draw=black, style={triang_lesssep}, doubled]
\tikzstyle{CPMkstate}=[fill=white, draw=black, style=kpoint, tikzit shape=rectangle, doubled]
\tikzstyle{CPMkstateconj}=[fill=white, draw=black, style=kpoint conjugate, tikzit shape=rectangle, doubled]
\tikzstyle{CPMkstateBIG}=[fill=white, draw=black, style=big kpoint, tikzit shape=rectangle, doubled]
\tikzstyle{CPMkeffect}=[fill=white, draw=black, style=kpoint adjoint, doubled]
\tikzstyle{CPMkeffectconj}=[fill=white, draw=black, style=kpoint transpose, doubled]
\tikzstyle{UHfB}=[fill=white, draw=black, style=triangdag, doubled, inner sep=-2pt]
\tikzstyle{leak}=[style=tinypoint, regular polygon rotate=-90]
\tikzstyle{leakfill}=[style=tinypoint, regular polygon rotate=-90, fill=black]
\tikzstyle{Z}=[style=dot, fill=green]
\tikzstyle{X}=[style=dot, fill=red]
\tikzstyle{black_dot}=[style=dot, fill=black]
\tikzstyle{white_dot}=[style=dot, fill=white]
\tikzstyle{qblack_dot}=[style=ddot, fill=black]
\tikzstyle{qwhite_dot}=[style=ddot, fill=white]
\tikzstyle{whitephase}=[style=wphase dot, fill=white]
\tikzstyle{qredphase}=[style=phase dot, fill=red]
\tikzstyle{qgreenphase}=[style=phase dot, fill=green]
\tikzstyle{had}=[style=hadamard, doubled]
\tikzstyle{box}=[style=hadamard]
\tikzstyle{classhad}=[style=hadamard]
\tikzstyle{antipode}=[style=anti]

\tikzstyle{dottededge}=[-, dotted]
\tikzstyle{double edge}=[-, style=doubled, draw=black, tikzit draw={rgb,255: red,18; green,168; blue,191}]
\tikzstyle{new edge style 0}=[<-]
\tikzstyle{new edge style 1}=[-, draw={rgb,255: red,223; green,66; blue,126}, fill={rgb,255: red,223; green,66; blue,126}]
\tikzstyle{new edge style 2}=[-, draw={rgb,255: red,14; green,188; blue,83}]
\tikzstyle{new edge style 3}=[<-, draw={rgb,255: red,223; green,66; blue,126}]

\usetikzlibrary{calc, positioning, shapes.geometric}
\usetikzlibrary{
	arrows,
	shapes,
	decorations,
	intersections,
	backgrounds,
	positioning,
	circuits.ee.IEC
	}

\theoremstyle{definition}
\newtheorem{defn}{Definition}
\theoremstyle{plain}
\newtheorem{thm}{Theorem}
\theoremstyle{plain}

\theoremstyle{plain}

\title{
Causality in Higher Order Process Theories}
\author{Matt Wilson \institute{Department of Computer Science, University of Oxford, Wolfson Building, Parks Road, Oxford, UK} \institute{HKU-Oxford Joint Laboratory for Quantum Information and Computation} \email{matthew.wilson@cs.ox.ac.uk} \and Giulio Chiribella
\institute{QICI Quantum Information and Computation Initiative, Department of Computer Science, The University of Hong Kong}
\institute{Department of Computer Science, University of Oxford, Wolfson Building, Parks Road, Oxford, UK} 
\institute{HKU-Oxford Joint Laboratory for Quantum Information and Computation}
\institute{Perimeter Institute for Theoretical Physics, 31 Caroline Street North, Waterloo, Ontario, Canada}
\email{giulio.chiribella@cs.ox.ac.uk}}

\def\titlerunning{
Causality in Higher Order Process Theories}
\def\authorrunning{M. Wilson \& G. Chiribella}

\begin{document}

\maketitle

\begin{abstract}
Quantum supermaps provide a framework in which higher order  quantum processes can act on lower order  quantum processes. In doing so, they enable  the definition and analysis of new quantum protocols and causal structures. Recently, key features of quantum supermaps were captured through a general categorical framework, which 
led to a framework of higher order process theories  (HOPT)  \cite{workinp}.  The HOPT framework models   lower and higher order transformations in a single unified theory, with its mathematical structure shown to coincide with the notion of a closed symmetric monoidal category.
Here we provide an equivalent construction of the HOPT framework from four simple axioms of process-theoretic nature.  We then use the  HOPT framework to establish connections between foundational features such as causality, determinism and signalling, alongside exploring their interaction with the mathematical structure of  $*$-autonomy.
\end{abstract}

\section{Introduction}
Traditional theories of physics focus on the evolution of states by means of physical processes. More recently, however, there has been a growing interest in an extended class of theories, where processes can themselves evolve under a higher level type of operations called {\em  supermaps} \cite{Chiribella2008QuantumArchitectureb,Chiribella2008TransformingSupermapsb,chiribella2009theoretical,Chiribella2013QuantumStructure, chiribella2013normal, Perinotti2017CausalComputations, Bisio2019TheoreticalTheory}. In quantum information, the development of  supermaps  stimulated  the study of new protocols involving the  manipulation of quantum processes and quantum causal structures \cite{chiribella2008optimal,bisio2010optimal,bisio2010information,Ebler2018EnhancedOrder,Chiribella2021IndefiniteChannels, Kristjansson2020ResourceCommunication, Salek2018QuantumOrders, Chiribella2012PerfectStructures, Coecke2014AResources, Araujo2014ComputationalGatesb, Guerin2016ExponentialCommunication, Abbott2020CommunicationChannels, Wilson2020ASwitches, Chiribella2020QuantumOrders, Sazim2020ClassicalChannels, Procopio2019CommunicationScenario, Procopio2020SendingOrders, Zhao2020QuantumOrderb, Selby2020CompositionalCoherenceb,Liu2035OperationalChannels, Liu2019ResourceErasure, WechsQuantumOrder, Bavaresco2020StrictDiscrimination, Bavaresco2021UnitaryStrategies, Dong2021Success-or-Draw:Computation, Dong2021TheOperations, Yokojima2021ConsequencesSuperchannels, Quintino2019ReversingOperations, GourDynamicalResources, Gour2021EntanglementChannel}. 
 In addition, quantum supermaps serve as a lens through which one can examine the kinds of causal structures which are compatible with quantum theory~\cite{hardy2007towards,chiribella2009beyond,Oreshkov2012QuantumOrder,Chiribella2013QuantumStructure,Castro-RuizDynamicsStructures}. 

Given the usefulness of the supermap framework for quantum theory, it is natural to try and extend it to more general physical theories.  A powerful approach for capturing the structural aspects of  physical processes is the framework of   {\em process theories} \cite{Coecke2017PicturingReasoning}, which emerged from research in the field of categorical quantum mechanics \cite{Abramsky2004AProtocols, Abramsky2003PhysicalProcessing, Heunen2019CategoriesTheory, Coecke2010QuantumPicturalism, Coecke2017PicturingReasoning}.  In this framework,   the notions of sequential and parallel composition of processes are placed at the forefront by adopting   
 the mathematical structure of a {\em symmetric monoidal category (SMC)} \cite{Lane1971CategoriesMathematician}. 
 The process theoretic framework,  often aided by its easy-to-use graphical language \cite{Coecke2010QuantumPicturalism, Coecke2017PicturingReasoning}, has led to categorical formalisation of the notions of entanglement \cite{Coecke2010TheEntanglement}, phase \cite{Coecke2011PhaseQubits}, complementarity \cite{Coecke2009InteractingDiagrammatics, Gogioso2019ADynamics}, causal/temporal structure \cite{Coecke2012TimeTime,Coecke2013CausalProcesses, Kissinger2019AStructure, Pinzani2020GivingOrders, Selby2017AAdjoint, DBLP:journals/corr/abs-1704-08086}, information extraction \cite{Coecke2007QuantumSums, Selby2017Leaks:More}, postivity \cite{Selinger2007DaggerAbstract}, dynamics \cite{Gogioso2019ADynamics}, and memory \cite{Carette2021GraphicalMemory}, and the interactions between them \cite{Selby2018ReconstructingPostulates, Tull2020ATHEORY, GalleyATheory}. 
 
 
 In a recent work \cite{Kissinger2019AStructure}, Kissinger and Uijlen approached the study of supermaps in the process theoretic framework.  Specifically, they built supermaps respecting causality constraints by starting from compact closed categories with sufficiently well-behaved environment structures as ambient categories. The higher order theories resulting from this construction were  named {\em higher order causal categories  (HOCCs)}, and were shown to be a special subclass of $*$-autonomous categories.
  On the other hand, one may ask which mathematical structure captures \textit{precisely} the notion of a higher order physical theory,  independently of the specific properties of the ambient category from which the higher  order transformations might be built, and independently of causality constraints.   An answer to this question was proposed  by the authors of the present paper, who introduced a categorical notion of supermap and, iteratively building on it, the notion of a {\em higher order process theory  (HOPT)} \cite{workinp}.  Mathematically,   HOPTs were shown to   coincide with {\em closed symmetric monoidal categories (CSMCs)}. The HOPT framework  permits the study of higher order theories in their own right, without reference to additional structures inherited by their particular means of construction, and prior to introduction of any notion of causality.
   In general, HOPTs  provide a broad arena for  studying  the interplay between physical axioms, operational features, and categorical  structures. 
    
   In this paper we present a simpler characterisation of the HOPT framework,  showing that the structure of closed monoidal category can be derived from four basic axioms about higher order processes.  The axioms revolve around the idea that the processes of type $A\rightarrow B$ must be in one-to-one correspondence with states of a higher order object $A\Rightarrow B$.      We then introduce the notion of a {\em tight} HOPT, as a HOPT $\mathcal{C}$ in which all objects are interpretable as types of higher order transformations acting on  a basic, first-order  theory $\mathcal{C}_1$ (in other words,  the objects of $\cal C$ are well-formed expressions built by combining the objects of ${\cal C}_1$ with the binary operations $\otimes$ and $\Rightarrow$).  We then show how the framework can be used to reason about higher order theories by establishing structural relations  between determinism,   properties of correlations, causality, signalling, and $*$-autonomy.  Specifically, we demonstrate that
   \begin{itemize}
    \item if ${\cal C}_1$ is causal  and all  single-state objects in $\cal C$ have no correlations with other objects, then for every pair of  objects $A$ and $A'$ in  ${\cal C}_1$ and any arbitrary object $X$ in $\cal C$ the tensor product system $(A\Rightarrow A')\otimes X$ does not permit signalling from system $A$ to system $X$. In other words, discarding $A'$ completely blocks the flow of information from $A$ to $X$. This result reproduces a key finding of \cite{Kissinger2019AStructure} with only reference to basic operational principles. 
     \item if every  object $A$ in  ${\cal C}_1$ is equivalent to its double dual   $ (A\Rightarrow I)  \Rightarrow I $, and   the tensor product  preserves equivalence with double duals,    then the HOPT $\cal C$ is  $*$-autonomous. 
     \item      if ${\cal C}_1$ is causal, $\cal C$  is  $*$-autonomous,    and all single-state objects in $\cal C$ have no correlations with other objects,  then for every pair of  objects $A$ and $A'$ in  ${\cal C}_1$ and every arbitrary  object $X$  in $\cal C$  the tensor product  object  $(A \Rightarrow A')\otimes X$ does not permit signalling from  $A \Rightarrow A'$ to $X$. In other words, the choice of a supermap acting on $A \Rightarrow A'$ cannot affect the marginal state of  $X$.
\end{itemize}
We also prove that the first and third results in the above list hold in a more general setting, where the HOPT $\cal C$ is not required to be tight.   In that setting, causality of $\mathcal{C}_1$ is replaced by the  requirements that $\cal C$ is  deterministic ({\em i.e.} has a unique scalar) and that objects $A$ and $A'$ are causal ({\em i.e.} they have a unique discarding operation \cite{chiribella2010probabilistic,chiribella2011informational,Coecke2013CausalProcesses,coecke2016terminality}). 

A potential avenue for future research is to generalise the work of \cite{Kissinger2019AStructure} to generate interesting examples of HOPTs beyond higher order causal categories, for example by generalising constructions to infinite dimensional process theories \cite{Cockett2018DaggerMechanics, Cockett2021ExponentialComplementarity, Gogioso2017Infinite-dimensionalMechanics, Gogioso2018TowardsMechanics, Gogioso2019QuantumMechanics} and time symmetric operational theories \cite{Hardy2021TimeTheories}, alongside including sectorial restrictions \cite{Vanrietvelde2021RoutedCircuits}. Furthermore there are connections to be explored with frameworks for causal inferential theories \cite{SchmidUnscramblingTheories}, string diagrams with open holes \cite{Roman2020CombFeedback, Roman2020OpenCalculus}, and extensions to the notion of a lambda calculus to quantum settings \cite{SelingerQuantumCalculus, Selinger2004TowardsLanguage, VanTonder2004AComputation, Zorzi2016OnPerspective}.
\section{Higher order process theories}  
\subsection{Introduction to higher order transformations}
In quantum theory,  
  deterministic state transformations are represented by quantum channels, that is, completely positive, trace-preserving linear maps acting on density matrices  \cite{Wilde2013QuantumTheory}.  In turn, quantum supermaps \cite{Chiribella2008QuantumArchitectureb,Chiribella2008TransformingSupermapsb,chiribella2009theoretical,Chiribella2013QuantumStructure, chiribella2013normal, Bisio2019TheoreticalTheory} describe deterministic transformations of quantum channels, and they are represented by linear maps on a suitable vector space of maps.
This notion of a higher order transformation acting on lower order transformations can be iterated indefinitely to construct an infinite hierarchy of  transformations of increasing complexity \cite{chiribella2009theoretical, Perinotti2017CausalComputations, Bisio2019TheoreticalTheory}.

In \cite{Kissinger2019AStructure}, Kissinger and Uijlen extended  the construction of quantum supermaps  to a large class of physical theories. Specifically, they provided a way to build a higher order theory $\mathbf{Caus}[\mathcal{P}]$ 
by imposing a causality axiom on a raw-material category $\mathcal{P}$, assumed to be compact closed.  The result of this construction was named as a ``Higher Order Causal Category" (HOCC), and was shown to be a special type of  $*$-autonomous category.
More recently, a broad categorical framework for theories of supermaps was introduced in \cite{workinp}, where we introduced the notion of Higher Order Process Theory (HOPT). HOPTs were shown to be mathematically equivalent to closed symmetric monoidal categories \cite{Lane1971CategoriesMathematician}, an important class of categories that contains $*$-autonomous categories (and so HOCCs) as a special case.


 Let us start with an informal summary of   the framework of  \cite{workinp}. Following \cite{Coecke2017PicturingReasoning}, in this framework a standard physical theory is modelled as a symmetric monoidal category (SMC) ${\cal C}_1$, with physical systems represented by objects and physical processes represented by morphisms between objects. When objects form a set we denote that set by $o({\cal C}_1)$ and for each pair of objects $A,B$ we denote the set of morphisms from $A$ to $B$ in $\mathcal{C}_1$ by $\mathcal{C}_1(A,B)$.  Symmetric monoidal structure of a theory ensures that it comes with a notion of parallel composition for objects and morphisms, represented by the symbol $\otimes$. Each SMC also comes equipped with a notion of empty space $I$ such that $A \otimes I$ is equivalent to $A$, the states of an object $B$ are then considered to be morphisms of the form $f:I \rightarrow B$. The category of supermaps over ${\cal C}_1$ is then another symmetric monoidal category $\cal C$, with the property that every pair of objects $A , B$  in  ${\cal C}_1$  is associated to an object type $A\Rightarrow B$ in $\cal C$ representing morphisms from $A$ to $B$ in ${\cal C}_1$, every process $f \in {\cal C}_1  (A,B)$   is then uniquely associated to a  state $\hat f \in {\cal C} (I  ,  A \Rightarrow  B)$.   We refer to the morphism $f \in  {\cal C}_1  (A,B)$ as a \textit{dynamic} process,  and to the state $\hat f  \in {\cal C}  (I,A\Rightarrow B)$ as the \textit{static} version of process $f$. Supermaps are considered to be the morphisms of $\cal C$, as a result they act on object types such as $A \Rightarrow B$. Axioms are given for two separate tensor products, one denoted $\otimes$ in which bipartite processes can have their parts plugged together in sequence or in parallel, and another denoted $\boxtimes$ which models the largest imaginable way to combine objects. We will see that the former product $\otimes$ is an abstract model for the non-signalling tensor product of \cite{Kissinger2019AStructure}. The latter product $\boxtimes$ on the other hand is analogous to the ``par" $\mathbf{\&}$ of \cite{Kissinger2019AStructure}. This manuscript will only be concerned with the former product.
 
 A theory $\mathcal{C}$ equipped with just the former product $\otimes$ contains its own supermaps if the above story holds with ${\cal C}_1=  {\cal C}$. Moreover,  the  lower and higher order levels within $\mathcal{C}$ are {\em linked}  if each object $A$ is isomorphic to the  object   $I\Rightarrow A$,  representing the processes from the unit object $I$ into $A$. When this condition is satisfied, $\mathcal{C}$ is called a HOPT.  Mathematically, HOPTs  can be characterised as closed symmetric monoidal categories (CSMCs) \cite{workinp}.  
 
 In the following subsection, we provide an alternative characterisation of HOPTs/CSMCs in terms of four simple axioms of process-theoretic nature. 

\subsection{Four axioms for higher order process theories}

Ref. \cite{workinp} argued that the appropriate mathematical structure for describing higher order process theories is the structure of a CSMC.  
In process-theoretic terms, CSMCs can be defined as follows: 
\begin{defn}
A CSMC $\mathcal{C}$ is an SMC  in which, 
 for every pair of objects $A$,$B$  in $\mathcal{C}$,  there exists an object $A \Rightarrow B$ in $\mathcal{C}$ and a morphism $\epsilon_{A \Rightarrow B}: (A \Rightarrow B) \otimes A \rightarrow B$ such that for every morphism $f:(C \otimes A) \rightarrow B$ there exists a unique morphism $\bar{f}:C \rightarrow (A \Rightarrow B)$ satisfying \begin{equation}\input{figs/evaldef.tikz} \quad = \quad \input{figs/evaldef_2.tikz}\end{equation} 
\end{defn}

The process $\bar{f}$ will often be referred to as the \textit{curried} version of the process $f$.  Currying is the key notion of a CSMC:  for each process $f: C \otimes A \rightarrow B$ there is a process $\bar{f}: C \rightarrow (A \Rightarrow B)$ which takes $C$ as an input and inserts it into the left hand input of $f$. 

Closed monoidal structure is powerful, but it is  unclear whether the existence of the curried version of each process should considered be a fundamental principle. Instead of assuming closed monoidal structure from the outset, we present four basic operational axioms that pin down the structure of a CSMC, and derive currying as a consequence.   The intention of the axioms is to capture the notion of a theory in which each process exists both in a static form, manipulable by higher order transformations within the same theory, as well as in a dynamic form in which such a process may be interpreted as actually \textit{happening} to a system.  

The axioms are imposed on a given process theory, mathematically described by an SMC $\cal C$.   Informally, the  axioms are as follows:  \begin{itemize}
    \item {\bf Axiom 1.}  For every pair of objects $A,B$ there exists an object $A\Rightarrow B$ such that for each process $f:A \rightarrow B$, there exists a unique state $\hat{f}:   I\rightarrow (A\Rightarrow B)$.
    \item  {\bf Axiom 2.}   There  exists  a higher order transformation which uses the static process $\hat f$ as a resource for implementing the dynamic process $f$. 
    \item {\bf Axiom 3.}   There exist higher order transformations which plug static processes together in sequence or in parallel. 
    \item {\bf Axiom 4.}  
    Every  state $\rho: I \rightarrow B$ is equivalent to its static representation $\hat{\rho}:  I  \rightarrow (I \Rightarrow B)$.
\end{itemize}
We now formally phrase the above axioms   in the language of process theories. Axiom 1 is already expressed formally. 
To formalise Axiom 2,  we introduce the notion of ``insertion of a process'': 
\begin{defn}
For a generic pair of objects $A$ and $B$, an {\em insertion} is  a process $\epsilon_{A,B}: (A \Rightarrow B) \otimes A \rightarrow B$ such that 
\begin{equation}
  \input{figs/eval_straight.tikz}
 \end{equation}
  for every $f:A \rightarrow B$.
\end{defn}
From here on  we will adopt the following notation \begin{equation}\input{figs/eval_notate.tikz}\end{equation}

Given any process  $f: C \rightarrow (A \Rightarrow B)$ that produces a {static} process in output, the insertion $\epsilon_{A , B}$ can be applied to the static output to make a new process.
Explicitly, the new process is obtained by applying the function \begin{equation}E_{A , B}^{C} : \mathcal{C}(C,A \Rightarrow B) \rightarrow \mathcal{C}(C \otimes A,B)\end{equation} defined by \begin{equation} E_{A , B}^{C}:: \quad  \textrm{  } \input{figs/eval_new_2.tikz} \end{equation}   
We say that  $\epsilon_{A , B}$ is \textit{completely injective} if the function $E_{A ,B}^{C}$ is injective  for every $C$. 

Physically,  since $\epsilon$ is  interpreted as usage of a process, it is natural to require that $\epsilon$ be completely injective.  The formal statement of Axiom 2 is that there exist a completely injective insertion $\epsilon_{A, B}$ for every pair of objects $A, B  \in  o ({\cal C})$.    Axioms 1 and 2 together imply that there is a bijective correspondence between the set of processes ${\cal C} (A,B)$ and the set of states ${\cal C}  (I, A\Rightarrow B)$.   Note that, however, there is an operational difference between static and dynamic processes:   a static process is   a resource for generating the corresponding dynamic process,  but the converse may not be true in general. 


Axiom 3 demands that sequential and parallel composition appear as higher order  processes that can be applied to static processes.  This idea is captured by the following definition: 
\begin{defn}
Let  $\mathcal{C}$ be a process theory equipped with a completely injective insertion $\epsilon_{A , B}$ for each pair of objects $A,B$, we say that $\cal C$ has \textit{basic manipulations} if for every triple $(A,B,C)$ and quadruple $(A,A',B,B')$ there exists processes $\circ_{ABC}$ and $\otimes_{AA'BB'}$ denoted \begin{equation}\input{figs/comp_pair.tikz}\end{equation} such that the following equations hold \begin{equation}
  \input{figs/sigmaevalnew.tikz} \quad \quad \quad \quad \input{figs/parplace2new_a.tikz}
 \end{equation}
\end{defn}
The above equations are reminiscent of equations used in causal inferential theories to derive composed states of knowledge from states of knowledge about individual processes \cite{SchmidUnscramblingTheories}. By inserting static processes and using the definitions of $\epsilon$, $\circ$ and $\otimes$ it is shown in Appendix A that $\circ$ and $\otimes$ implement sequential and parallel composition of static processes, respectively. We will from here on adopt a special notation for the static representation of the identity $\hat{id_{A}}: I \rightarrow (A \Rightarrow A)$: \begin{equation}
    \input{figs/identity_notation.tikz}
\end{equation}

Finally,  Axiom 4   postulates an equivalence between each object $A$ and the corresponding  object  $(I \Rightarrow A)$.    Formally we require the insertion $\epsilon_{I,A}: (I \Rightarrow A) \otimes I  \rightarrow A$  to be an   isomorphism for every object $A \in  o(\cal C)$. In string diagram language this is phrased by asking for a process $\eta_{A}$ such that 
\begin{equation}
    \input{figs/eval_well_point.tikz}
\end{equation}
We will in general adopt an aesthetic convention of notating with small boxes or circles those processes which are canonical, in other words, those who's existence follows from the axioms of a higher order process theory alone. 

\begin{defn}[Higher order process theory]
A {\em higher order process theory (HOPT)} is an SMC  $\mathcal{C}$ equipped with a completely injective insertion $\epsilon_{A , B}$ for every pair of objects $A$, $B$ such that
\begin{itemize}
    \item $\mathcal{C}$ has basic manipulations
    \item For each $A$ the map $\epsilon_{I ,A}$ is an isomorphism
\end{itemize}
\end{defn}
As it  turns out, these conditions are equivalent to providing a closed symmetric monoidal structure:  
\begin{thm}[HOPTs =  CSMCs]
An SMC  $\mathcal{C}$  is a HOPT  if and only if $\mathcal{C}$ is a CSMC. 
\end{thm}
\begin{proof}
Given in Appendix B. The key idea is that the curried version of a generic process can be constructed from its static version using the inverse of the insertion $\epsilon_{I,A}$ along with basic manipulations. 
\end{proof}

\subsection{Tight higher order process theories}
In a generic HOPT, there is no explicit distinction between higher levels and lower levels. In particular, there is no specification of a first-order  physical theory $\mathcal{C}_1$ on which the higher order processes of  $\mathcal{C}$ are based.
We now add such a specification by requiring the existence  of a first-order theory $\mathcal{C}_1$ inside of $\mathcal{C}$, such that all of the processes in $\mathcal{C}$ can be interpreted as manipulations of processes built from $\mathcal{C}_1$. The definition presented here is a special case of a more general notion of a higher order theory $\mathcal{C}$ containing a first-order theory $\mathcal{C}_1$ introduced in \cite{workinp}.

A {\em full sub-process theory} $\mathcal{S}$ of a process theory $\cal C$ is a symmetric monoidal subcategory 
of $\mathcal{C}$ such that for any pair of objects $A,B$ in $\mathcal{S}$ the processes from $A$ to $B$ in $\mathcal{S}$ are \textit{all} of the processes from $A$ to $B$ in $\mathcal{C}$.  
\begin{defn}
A \textit{tight} HOPT is a pair   $(\mathcal{C}, {\cal C}_1)$ where 
\begin{itemize}
    \item $\mathcal{C}$ is a HOPT and $\mathcal{C}_1$ is a full sub-process theory of $\mathcal{C}$
    \item The objects of $\mathcal{C}$ are generated by combining the objects of $\mathcal{C}_1$ with the binary operations $\otimes$ and $\Rightarrow$, that is, they are given by the algebra $ o(\mathcal{C}_1) \, \thickvert \, \otimes_{\mathcal{C}}  \, \thickvert \,  \Rightarrow_{\mathcal{C}}$.
\end{itemize}
\end{defn}

For a tight  HOPT $(\mathcal{C},  {\cal C}_1)$, we will see in section $5$ that the closed monoidal structure imposes constraints that are strong enough to allow a lifting of  certain properties from the objects of   $\mathcal{C}_1$  to all objects in ${\cal C}$.

\section{String diagram toolbox}
We now develop a graphical representation of some basic notions in higher order physics, such as the notions  of combs and acyclic causal structures.

\paragraph*{Combs}  Quantum combs~\cite{Chiribella2008QuantumArchitectureb,chiribella2009theoretical} represent quantum circuits with a set of open holes in which quantum channels can be inserted. In the categorical framework, the canonical morphisms of a HOPT $\mathcal{C}$ give formal meaning to such circuits of the theory with open holes: in the HOPT framework, a comb is simply represented by a special type of  morphism in  $\mathcal{C}$. For example, a comb with a single hole for a process of type $A\to A'$ (left-hand side of the following diagram)  is represented by  a morphism containing an insertion of the static type $A\Rightarrow A'$  (right-hand side of the following diagram) 
\begin{equation}
    \input{figs/evalcomb_a_bold.tikz}
 \end{equation}
The sign $\cong$ denotes a correspondence between an informal picture on the left hand side and a morphism used to represent it on the right hand side.     

Every comb defines a supermap, whereby the processes inserted in the empty holes are transformed into  new processes. In the static picture,  the action of this  supermap  is  generated by the basic operations of parallel and sequential composition. For instance,  the supermap $s_{f,g}$  corresponding to the comb in the above diagram can be decomposed as 
\begin{equation}
  \input{figs/sigmafunctor.tikz}
 \end{equation}

\paragraph*{Acyclic causal structures}
The canonical processes of any HOPT are sufficient to define the insertion of processes into the vertices of any arbitrary directed acyclic graph. This scenario can be represented by formal diagrams in the HOPT, which  may prove useful for reasoning about information theoretic protocols involving the  agents who perform operations at the nodes of a network.
 \begin{defn}[Circuit skeleton]
A circuit skeleton in a tight HOPT $\mathcal{C}$ is a circuit built only from insertion processes $\epsilon_{A , B}$ with $A,B \in o(\mathcal{C}_{1})$.
\end{defn}
An example of a circuit skeleton is the following process in which thin wires are used to represent objects of $\mathcal{C}_1$:
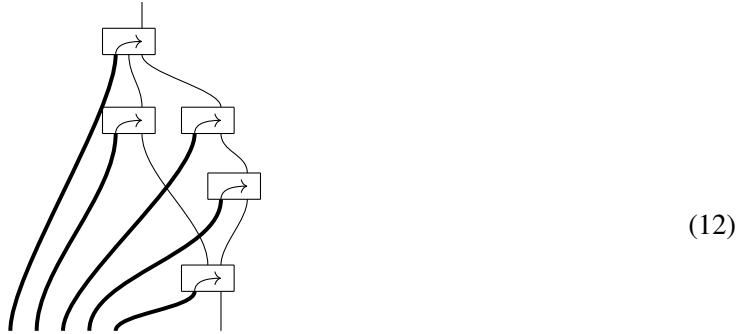
\begin{equation}
  \input{figs/causalskeleton_a.tikz}
 \end{equation}

When $\mathcal{C}_1$ is causal a circuit skeleton can be interpreted as a raw causal structure of nodes into which physical processes can be freely inserted.     Note that, more broadly,  circuit skeletons could also be used in general non-tight HOPTs by allowing insertion processes $\epsilon_{A,B}$ with arbitrary systems $A, B \in \cal C$.   

\paragraph*{Dualising processes}
Intuitively, it should be possible to view a state of object $A$ as an ``effect on the effects on $A$,'' that is, as a transformation that maps effects in $A\Rightarrow I$ into scalars (i.e.~elements of $I$).    In a  CSMC, the embedding of $A$ into $(A\Rightarrow I)  \Rightarrow I$  is implemented by a process  
\begin{equation}d_A: A~ \rightarrow ~\big[(A \Rightarrow I) \Rightarrow I\big] \end{equation} 
uniquely defined by the following condition 
 \begin{equation}\input{figs/staraut_a.tikz}\end{equation}
  The existence of the unique morphism $d_{A}$ is well-known, and a proof is given in Appendix C. 

In the following, we will call $d_A$  the \textit{dualising process} for system $A$. It is natural to require that the dualising process maps distinct states of $A$ into distinct states of  $(A\Rightarrow I)\Rightarrow I$.   If this injectivity property holds for every  object $A\in  o({\cal C})$, we say that the HOPT $\cal C$ has {\em injective dualisation}. An example of a HOPT with  injective dualisation is a  theory  with ``enough effects,''  in the following sense:  
\begin{defn}[Enough effects]
A HOPT    $\mathcal{C}$ has   {\em enough effects}  if for every object  $A\in  o({\cal C})$ and for every pair of states $\rho  , \sigma  \in  {\cal C}  (I,  A)$,  the condition $\forall e  \in  {\cal C}  (A, I):  ~ \textrm{ }  e  \circ   \rho  =  e\circ \sigma  \, , $ implies $\rho  =  \sigma$. \end{defn} A proof that enough effects imply injective dualisation is given in Appendix D. Whilst in general the dualising process $d_A$ may not be an isomorphism, if every $d_A$ is indeed an isomorphism, then $\cal C$ is $*$-autonomous:  
\begin{defn}[$*$-autonomous category with global dualising object $I$]
A closed symmetric monoidal category  $\mathcal{C}$ is  {\em $*$-autonomous with global dualising object $I$} if $d_A$ is an isomorphism for every $A\in  o(\cal C)$. \end{defn}    
In fact the above is a special case of the more refined notion of an \texttt{ISOMIX} \cite{Cockett1997PROOFCATEGORIES} category. Later in this paper we will discuss the relation between the special case of $*$-autonomous HOPTs, and the HOCCs of Ref. \cite{Kissinger2019AStructure}.  

\paragraph*{Lifting processes on states to processes on effects}
In a CSMC, it is  possible to show that each state of a system  $A \Rightarrow B$  can be converted into a state of the system $(B \Rightarrow I) \Rightarrow (A \Rightarrow I)$ representing a process from $(B \Rightarrow I)$ to $(A \Rightarrow I)$. The conversion \begin{equation}T_{AB}:\big(A \Rightarrow B\big) \rightarrow  \big[ (B \Rightarrow I) \Rightarrow (A \Rightarrow I)\big]\end{equation} termed the \textit{lifting process} is defined by the following condition \begin{equation}\input{figs/adjoint_a.tikz}\end{equation}
The existence of the lifting process $T_{AB}$ is proven in the Appendix C.
\paragraph*{Static currying}
Every state of type $C \Rightarrow (A \Rightarrow B)$ defines a process $C \rightarrow (A \Rightarrow B)$, which in turn defines a process of type $(C \otimes A) \rightarrow B$ and so a state of type $(C \otimes A) \Rightarrow B$. The correspondence between states of $C \Rightarrow (A \Rightarrow B)$ and $(C \otimes A) \Rightarrow B$ is clearly one-to-one.  Furthermore, it is possible to show that this correspondence is implemented by an isomorphism  \begin{equation}\phi: \big[ C \Rightarrow (A \Rightarrow B)\big] ~\rightarrow ~ \big[(C \otimes A) \Rightarrow B\big]\end{equation} defined by \begin{equation}\input{figs/phidefined_a.tikz}\end{equation}  A short diagrammatic proof that  $\phi$ is an isomorphism is provided in Appendix E.  Alternatively, the isomorphism property of $\phi$ can be derived from the Yoneda lemma.



\section{Causality in higher order process theories}
We now introduce causality into the picture.
In a probabilistic setting, the causality axiom states that the probability of outcomes obtained at a certain step of a circuit cannot depend on the choice of operations performed at later steps  \cite{chiribella2010probabilistic,chiribella2011informational,Chiribella2016QuantumPrinciples}. This axiom is equivalent to the condition that there exists a unique deterministic effect, this unique effect is typically written with the following ``ground" symbol: \begin{equation}\begin{tikzpicture}[scale=0.6, {every node/.style}={scale=0.6}]
	\begin{pgfonlayer}{nodelayer}
		\node [style=none] (0) at (0, -0.25) {};
		\node [style=discarding] (1) at (0, 0.75) {};
		\node [style=none] (2) at (0.5, 0) {$A$};
	\end{pgfonlayer}
	\begin{pgfonlayer}{edgelayer}
		\draw (1) to (0.center);
	\end{pgfonlayer}
\end{tikzpicture}
\end{equation}In the categorical  setting, if one restricts their attention to the category of deterministic processes, causality is the statement that the monoidal unit $I$ is terminal \cite{Coecke2013CausalProcesses,coecke2016terminality}.

\subsection{Causality and determinism}
To formulate causality in a HOPT,  it is convenient to first define the notion of determinism. 
  In a deterministic theory, there should only be one scalar, which represents certainty. 
\begin{defn}[Deterministic process theory]
A process theory $\mathcal{C}$ is {\em deterministic} if it contains  only one scalar, that is, if $|{\cal C} (I,I)|=1$.  The unique scalar in a deterministic theory is denoted by $1$.
\end{defn}
HOCCs provide an instance of deterministic HOPTs.
  
\begin{defn}[Causal object/theory]
An object  $A$   is  {\em causal}  if it has only one effect, that is, if $|{\cal C}  (A,I)|  = 1$.  A process theory  $\cal C$ is {\em causal} if all the objects $A\in  o (\cal C)$ are causal.
\end{defn}

  Note that every causal theory is automatically deterministic.   In the higher order setting,  it is interesting to study tight HOPTs $(\mathcal{C},\mathcal{C}_1)$ in which the first-order theory ${\cal C}_1$ is causal. In this case, it is immediate to see that  $\cal C$ is  deterministic. Notice that, however, it does not make much sense to study the scenario in which an entire theory $\cal C$ is  causal, because any such theory is trivial under the reasonable assumption that the dualisations are injective: 
\begin{thm}
A HOPT $\cal C$ with injective dualisation  is causal if and only if it is trivial, that is, if and only if $|{\cal C}  (A,B)|=1$ for all objects $A,  B  \in  o(\cal C)$.  
\end{thm}
 \begin{proof} If  $|{\cal C}  (A,B)|=1$ for every pair of objects $A,  B$, then $\cal C$ is trivially causal.  Conversely, assume that  $\cal C$ is a causal HOPT.  Then, for a generic object $A\in o ({\cal C})$,  pick two generic states $\rho, 
 \sigma \in  {\cal C}  (I,  A)$, and consider the states  $d_A  \circ \rho$ and $d_A  \circ \sigma$ of $(A\Rightarrow I)  \Rightarrow I$.   These states are in one-to-one correspondence with effects on system $A\Rightarrow I$.   Since the theory is causal, system $A\Rightarrow I$ has only one effect, and therefore we must have  $d_A  \circ \rho  = d_A  \circ \sigma$. Since the  dualisation $d_A$ is injective, we have $\rho  =  \sigma$.  Hence, we conclude that  system $A$ has only one state.  More generally, for a generic pair of objects $A,B \in  o ({\cal C})$, the  morphisms  of type $A \rightarrow B$ are in one-to-one correspondence with the states of  $A\Rightarrow B$, and therefore one has $|{\cal C}  (A,B)|=1$.   \end{proof}

Note that the above theorem holds in particular when the category $\cal C$ is $*$-autonomous with $I$ the global dualizing object. In summary, the relevant scenario for causality in HOPTs is the one in which a sub-theory ${\cal C}_1$ is causal, while the entirety of $\cal C$ is only deterministic. We conclude the section by showing that, if $\cal C$ is deterministic, a simple sufficient condition for an object to be causal is that it has ``enough states,'' in the following sense:  
\begin{defn}[Enough states]
An object $A$ has enough states if for every object $X$ and for every pair of processes $f,g:A \rightarrow X$ 
\begin{equation}
    f = g \iff \forall  \rho  \in  {\cal C}  (I, A): \textrm{ } f\circ \rho = g \circ \rho
\end{equation}
\end{defn}
In the axiomatic framework of \cite{Chiribella2016QuantumPrinciples,mauro2017quantum}, this property can be shown to follow from the condition of local distinguishability, also known as local tomography \cite{araki1980characterization,wootters1990local,Hardy2001QuantumAxioms,dariano2006how,barnum2007generalized,Barrett2007InformationTheories,chiribella2010probabilistic,hardy2011reformulating}.

In any deterministic HOPT if an object $A$ has enough states then it must be causal, i.e. there can be only one effect $A \rightarrow I$. Any two effects $e_1,  e_2  \in  {\cal C}  (A,I)$ satisfy the condition $  e_1  \circ \rho   =  1   =  e_2\circ \rho $ for every state $\rho \in  {\cal C}  (I,A)$, and therefore the ``enough states'' condition implies $e_1  = e_2$.

\subsection{The no-signalling tensor product}\label{subsec:nosigtens}  

An important insight of Ref. \cite{Kissinger2019AStructure} is that the tensor product in a higher order causal category does not allow for signalling between tensor factors of process types between causal objects.  More specifically, Ref. \cite{Kissinger2019AStructure} showed that for any first-order objects $A,B,A',B'$ of a HOCC the type $(A \Rightarrow A') \otimes (B \Rightarrow B')$ represents the space of non-signalling channels, for which the output $A'$ has no dependence on the input $B$ and the output $B'$ has no dependence on the input $A$. This notion can be expressed in the language of HOPTs whenever each of $A,B,A',B'$ has a unique effect: a state $ f : I  \rightarrow (A \Rightarrow A') \otimes (B \Rightarrow B') $ represents a non-signalling channel if  there exist (dynamic) processes $f_A:  A\rightarrow A'$ and $f_B:  B \rightarrow B'$ satisfying: \begin{equation}\input{figs/beamer_nonsig_B.tikz} \quad \quad \quad \input{figs/beamer_nonsig_A.tikz}\end{equation}   
An interesting question is whether the above no-signalling property of the tensor product in a HOCC can be derived through operational principles imposed on a general HOPT.

We now introduce a condition that implies this no-signalling property of the tensor product. The condition is that objects with a single state cannot form non-separable joint states with other objects. Intuitively,  if  a  joint state of objects $X$ and $Y$ is interpreted as representing correlations between the states of $X$ and $Y$, it should not be possible to correlate any auxiliary object $X$ with a single-state object $Y$. This intuition motivates the following definition: 
\begin{defn}[No correlation with a single-state object]
A process theory $\mathcal{C}$ has {\em no correlations with single-state objects} if,  for any object $Y$ with $|\mathcal{C}(I,Y)| = 1$ and any object $X\in  o(\cal C)$, every state $\rho: I \rightarrow X \otimes Y$ is of the product form  $\rho = \rho' \otimes \pi$ with $\rho' \in \mathcal{C}(I,X)$ and $\pi \in \mathcal{C}(I,Y)$
\begin{equation}
  \input{figs/trivialdegree.tikz}
 \end{equation}
\end{defn}

The above condition is satisfied by all HOCCs as defined in \cite{Kissinger2019AStructure}:
\begin{thm}
Every HOCC is a HOPT with no correlations with single-state objects.
\end{thm}
\begin{proof}
A minor generalisation of lemma 6.1 of \cite{Kissinger2019AStructure}, given for completeness in Appendix F. 
\end{proof}
The condition of ``no correlation with single-state objects"   was crucial  to proving that $(A \Rightarrow A') \otimes (B \Rightarrow B')$ represents a non-signalling channel in \cite{Kissinger2019AStructure}. In that context, the statement followed from a specific decomposition of supermaps, as  open circuits of causal processes.    Here, instead, we take the ``no correlation with single-state objects" as a basic operational condition. 

We now show that, if there is no correlation with single-state objects, then the tensor product has a no-signalling property. For a given process, non-signalling is defined as follows:
\begin{defn}[Non-signalling process]
A process  $m:   A \rightarrow  A' \otimes X $ in a deterministic process theory  is {\em non-signalling}  from $A$ to $X$  if for every effect $\pi_{A'}: A' \rightarrow I$ there exists an effect $\pi_{A}:A \rightarrow I$ and a state $\rho  :   I\rightarrow X$ such that \begin{equation}\input{figs/definition_non_sig.tikz}\end{equation}
\end{defn}
The definition expresses the idea that when $A'$ is discarded (in any way) no signal may reach $X$ from $A$.  Note that, in principle, the definition still allows for a notion of signalling from $A'$ to $X$, because in general the state $f'$ of $X$ could depend on the effect  $\pi_{A'}$ used for discarding. Note, however, that signalling from $A'$ to $X$ is not possible if system $A'$ is causal, because in that case the effect $\pi_{A'}$ is unique. In the following, we will restrict our attention to the case where both systems  $A'$ and $A$ are causal.


\begin{thm}[Non-signalling processes]
Let $\cal C$ be  a  deterministic HOPT with no correlations with single-state objects,  $A,A'$ be two causal objects in $\cal C$, and  $X\in  o(\cal C)$ be an arbitrary object.   Then, for every state  $f: I \rightarrow   (A \Rightarrow A') \otimes X$ the process $m$ defined by: \begin{equation}\input{figs/sigdotbox.tikz}\end{equation}  is non-signalling from $A$ to $X$. 
\end{thm}
\begin{proof}
As in \cite{Kissinger2019AStructure}, the core of the proof is the  ``no correlation with single-state objects"  property.    In the proof, this property is applied to  the object  $A \Rightarrow I$, which is a single-state object because $Hom(I,A \Rightarrow I) \cong Hom(I \otimes A,I) \cong Hom(A,I) \implies |Hom(I,A \Rightarrow I)| = 1$. The discarding effect can as a result be pulled through the entire process 
\begin{equation}
  \input{figs/provenonsig1_a.tikz}
 \end{equation}
The composition of $\hat f$ with the unique discarding effect on $A'$ at the bottom of the diagram gives a state of type $(A\Rightarrow I) \otimes X$, and so  ``no-correlation with single-state objects'' implies that such a state separates as the unique discarding state on $(A\Rightarrow I)$ and a state $f'$ on $X$:
 \begin{equation}
  \input{figs/provenonsig2_a.tikz}
 \end{equation}
\end{proof}
The above immediately entails the fact that states of type $f : I \rightarrow (A \Rightarrow A') \otimes (B \Rightarrow B')$ represent non-signalling channels (when $A,A',B,B'$ are causal) in the sense of \cite{Kissinger2019AStructure}, since for such a state $f$ then \begin{equation}
  \input{figs/beamer_proof.tikz}
 \end{equation}
The broad takeaway is that it is the causality of an object $A$ that prevents it from signalling to another object that it is in parallel with.

\subsection{Tensor product processes vs bipartite processes}  

For arbitrary objects $A,A',B,B'$, there is a parallel composition process from the tensor product object $(A\Rightarrow A') \otimes (B\Rightarrow B')$ to the space of bipartite processes $(A\otimes B)  \Rightarrow  (A'\otimes B')$.     
But can this morphism be an isomorphism? In other words, can the tensor product of processes of type $A\rightarrow A'$ and processes of type $B\rightarrow B'$  yield the full set of processes of type $(A\otimes B) \rightarrow  (A' \otimes B')$? Here we show that the answer is negative when $A' =  B$ and $B'  =  A$, since in this case the existence of a $\texttt{SWAP}$ process can be leveraged.
\begin{thm}
Let $\mathcal{C}$ be a deterministic HOPT with no interaction with single-state objects. If $A$ and $B$ are causal and 
\begin{flalign}
\input{figs/tensor.tikz} : (A \Rightarrow B) \otimes (B \Rightarrow A) \rightarrow (A \otimes B) \Rightarrow (B \otimes A)
 \end{flalign}
is an isomorphism, then $A$ and $B$ are single-state objects. 
\end{thm}
\begin{proof}
Given in Appendix G. The key idea is that the set of processes from $A\otimes B$ to $B\otimes A$ contains the swap of objects $A$ and $B$, and requiring the swap to be no-signalling implies that $A$ and $B$ have only one state each.
\end{proof}

\section{The emergence of $*$-autonomy}
An important difference between the HOPTs studied in this paper and the HOCCs of  \cite{Kissinger2019AStructure} is that the latter are not just closed monoidal, but also  $*$-autonomous, since they are equipped  with isomorphisms of the form $((A \Rightarrow I) \Rightarrow I) \cong A$ for every object $A$.  Here we explore the lifting of $*$-autonomy from lower to higher orders by showing that for a tight HOPT $(\mathcal{C},\mathcal{C}_{1})$, the property of $*$-autonomy can be lifted from the first-order theory $\mathcal{C}_{1}$ to the entire higher order theory $\cal C$ whenever the tensor product is sufficiently well behaved.  

\begin{defn}[Equivalence of double duals]
An object $A$ in a HOPT $\mathcal{C}$ is {\em canonically equivalent to its double dual} if $d_{A} :A ~\rightarrow ~ \big [(A \Rightarrow I) \Rightarrow I\big]$ is an isomorphism.
\end{defn}
Such an isomorphism forces states on $A$ to be nothing other than the effects on effects for $A$, as is the case in finite dimensional quantum systems.   This equivalence can be expressed more generally as  a symmetry between the dynamics on states and the dynamics on effects, such as the symmetry between the Schr\"odinger picture and the Heisenberg picture in quantum theory.  
\begin{defn}[Adjoint dynamics]
A HOPT $\mathcal{C}$ has {\em adjoint dynamics} between $A$ and $B$ if the morphism $T_{AB}:(A \Rightarrow B)~ \rightarrow ~  \big [(B \Rightarrow I) \Rightarrow (A \Rightarrow I)\big]$ is an isomorphism.
\end{defn}
Adjoint dynamics expresses the condition that the processes that may be applied to states are precisely those that may be applied to effects.
\begin{thm}
Let $\mathcal{C}$ be a HOPT, the following statements are equivalent. 
\begin{itemize}
    \item For all $A,B \in o(\mathcal{C})$ the HOPT $\mathcal{C}$ has adjoint dynamics between $A$ and $B$
    \item For all $B \in o(\mathcal{C})$ the HOPT $\mathcal{C}$ has adjoint dynamics between $I$ and $B$
    \item Every $B \in o(\mathcal{C})$ is canonically equivalent to its double dual in $\mathcal{C}$
\end{itemize}
\end{thm} 
\begin{proof}
Given in Appendix H. 
\end{proof}
Given  two systems $A$ and $B$ that are canonically equivalent to their double duals,  it is natural to ask whether equivalence is preserved by the binary operations  $(-\otimes-)$ and $(-\Rightarrow-)$, in the following sense: 
\begin{defn}[Preservation of equivalence of double duals]
A binary operation $\odot: o({\cal C}) \times o({\cal C}) \rightarrow o({\cal C})$ {\em preserves equivalence of double duals} if $d_{A \odot B}$ is an isomorphism  whenever $d_A$ and $d_{B}$ are isomorphisms.  
\end{defn}
We now show that the preservation of equivalence by the tensor product $\otimes$ is enough to guarantee preservation of the equivalence by the higher order composition $\Rightarrow$: 
\begin{thm}[Lifting canonical isomorphisms]
For every HOPT $\mathcal{C}$, if $(-\otimes-)$ preserves equivalence of double duals then $(- \Rightarrow -)$ preserves equivalence of double duals.
\end{thm}
\begin{proof}
Given in Appendix I.
\end{proof}
For every  tight   HOPT  $({\cal C},  {\cal C}_{1})$, a crucial consequence of the above theorem is that $*$-autonomy lifts from first-order to higher orders,  provided that the tensor product preserves equivalence with double duals: 
\begin{thm}
Let $(\mathcal{C},\mathcal{C}_{1})$ be a \textit{tight} HOPT. If 
\begin{itemize}
    \item for all objects $A \in \mathcal{C}_1$ the canonical morphism $d_A: A \rightarrow [I \Rightarrow (I \Rightarrow A)]$ is an isomorphism, and 
    \item the monoidal product $\otimes:\mathcal{C} \times \mathcal{C} \rightarrow \mathcal{C}$ preserves equivalence of double duals, 
\end{itemize}
then $\mathcal{C}$ is $*$-autonomous with dualising object $I$.
\end{thm}
\begin{proof}
Follows immediately from the fact that the objects of $\cal C$ are generated from the objects of ${\cal C}_1$ through the operations $\otimes$ and $\Rightarrow$.  
\end{proof}
\section{A stronger no-signalling property}
We conclude the paper by showing a strengthening of the no-signalling property shown in subsection~\ref{subsec:nosigtens}. There we saw that in a deterministic theory with no correlations with single-state objects, the states of type $(A \Rightarrow A') \otimes X$ represent processes which are non-signalling from $A$ to $X$ whenever $A$ and $A'$ are causal objects. We now show that, in the presence of equivalence to double duals, this no-signalling property can be strengthened: the tensor product $(A\Rightarrow A')\otimes X$ is no-signalling from the whole system $(A\Rightarrow A')$ to $X$.
\begin{defn}
An object $Y$ in a deterministic process theory   $\mathcal{C}$ has {\em no-signalling states}   if for every object $X$ and every bipartite state $m: I \rightarrow Y \otimes X$ there exists a state $m': I \rightarrow X$ such that for every $\Pi: Y \rightarrow I$ 
\begin{equation}
    \input{figs/ignore1.tikz}
 \end{equation}
\end{defn}
In other words an object $Y$ has no-signalling states if the choice of effect for discarding object $Y$ in a bipartite object $X\otimes Y$ does not affect the marginal state of system $X$.  
\begin{thm}
Let $\mathcal{C}$ be a deterministic HOPT with no correlations with single-state objects. If
\begin{itemize}
    \item $\otimes$ preserves equivalence with double duals, and
    \item $A$ and $A'$ are causal and canonically equivalent to their double duals,
\end{itemize}
then the object $(A \Rightarrow A')$ has no-signalling states. 
\end{thm}
\begin{proof}
Given in Appendix J.
\end{proof}
The theorem shows that, no matter which supermap is applied on the system $A\Rightarrow A'$, and no matter the way a system is discarded, the state of any other system in parallel will be unaffected. Indeed, for every pair of processes $S:   (A\Rightarrow A')  \rightarrow Y$ and $T:   (A\Rightarrow A')  \rightarrow Z$, and every pair of effects $e:   Y\rightarrow I$  and $k:  Z\rightarrow I$, one has 
\begin{equation}
    \input{figs/ignore_demo_2.tikz}
 \end{equation}
In other words, the choice of a supermap on system $A\Rightarrow A'$ cannot signal to any other system $X$. This can be seen as a generalised causality condition for circuits of processes within a HOPT.

\section{Conclusions}
We presented   HOPTs/CSMCs   as an operationally motivated  framework for higher order physics. By using the diagrammatic gadgets which come with a HOPT, we recovered signalling restrictions between process wires as a consequence of simple principles. We demonstrated that for a sufficiently tame notion of parallel composition the defining condition of $*$-autonomy (with global dualising object $I$) lifts from a first-order theory to its entire higher order theory. Following on from this, we showed that HOPTs with the above notion of $*$-autonomy satisfy a stronger causality condition, namely that a supermap on first-order processes cannot be used to signal to other factors of a tensor product. We hope that the definition of HOPTs will serve as a tool to guide the exploration of new structures arising in higher order physical theories.

\section{Acknowledgments}
MW would like to thank B Coecke, A Vanrietvelde, H Kristjánsson, J Hefford, A Kissinger, V Wang, J Selby, and G Boisseau for useful conversations. This work is supported by the Hong Kong Research Grant Council through grant 17300918 and though the Senior Research Fellowship Scheme SRFS2021-7S02, by the Croucher Foundation, by the John Templeton Foundation through grant 61466, The Quantum Information Structure of Spacetime (qiss.fr).  Research at the Perimeter Institute is supported by the Government of Canada through the Department of Innovation, Science and Economic Development Canada and by the Province of Ontario through the Ministry of Research, Innovation and Science. The opinions expressed in this publication are those of the authors and do not necessarily reflect the views of the John Templeton Foundation. MW gratefully acknowledges support by University College London
and the EPSRC Doctoral Training Centre for Delivering Quantum Technologies.

\bibliographystyle{eptcs}
\bibliography{refs_18_08_2021}

\appendix

\section{Well behaviour of sequential and parallel composition supermaps}
We check that basic manipulations behave as expected whenever they exist in a process theory.
\begin{thm}
Let $\mathcal{C}$ be a process theory equipped with a completely injective insertion $\epsilon_{A , B}$ for each pair of objects $A,B$ and with \textit{basic manipulations} for all objects, then it follows that for each $f,g$ and manipulation $\otimes$ or $\circ$: \begin{equation}\input{figs/appendix_s.tikz} \quad \quad \quad \quad \input{figs/appendix_p.tikz}\end{equation}
\end{thm}
\begin{proof}
For sequential composition note that \begin{equation}\input{figs/appendix_comp.tikz}\end{equation} and so the result is entailed by complete injectivity of the insertion, which allows the removal of insertions whilst preserving equality of diagrams. The proof for the parallel composition supermap is almost identical. 
\end{proof}
Complete injectivity also implies an associativity property of the sequential composition maps. Namely noting that:
\begin{equation}\input{figs/spider_appendix_1.tikz}\end{equation}
it follows that,
\begin{equation}\input{figs/spider_appendix_3.tikz}\end{equation}
this associativity property means that a $3$ input sequential composition map can be written unambiguously as 
\begin{equation}\input{figs/spider_appedix_4.tikz}\end{equation}
and similarly for $n$-input sequential composition processes. Furthermore each sequential composition of type $(A \Rightarrow B) \otimes (A \Rightarrow A) \rightarrow (A \Rightarrow B)$ has the static version of the identity as its right-unit, meaning that the following equation holds: \begin{equation}\input{figs/unit_proof_2.tikz}\end{equation} indeed this follows from noting that \begin{equation}\input{figs/unit_proof.tikz}\end{equation} and again using complete-injectivity. An almost identical proof can be used to show that the static identity $\hat{id_{A}}$ acts as a left-unit for the sequential composition of type $(A \Rightarrow A) \otimes (A \Rightarrow B) \rightarrow (A \Rightarrow B)$. The properties of associativity and unitality also entail that the assignment $f \Rightarrow g$ given by pre-composition with $f$ and post-composition with $g$:
\begin{equation}\input{figs/appendix_bifunctor.tikz}\end{equation}
is a bifunctor, meaning that $(f \Rightarrow g) \circ (f' \Rightarrow g') = (f' \circ f) \Rightarrow (g \circ g')$ and that identities are preserved. The above equation can be demonstrated to be true by witnessing two equal interpretations of the same $5$ input diagram:
\begin{equation}\input{figs/appendix_bifunctor_2.tikz}\end{equation}
It follows from the above that whenever $f$ is an isomorphism (meaning that it has both a left and a right inverse) and $g$ is an isomorphism then $f \Rightarrow g$ is an isomorphism.
\section{Equivalence between higher order process theories and closed monoidal categories}

\begin{thm}[HOPTs =  CSMCs]
A symmetric monoidal category $\mathcal{C}$  is a HOPT  if and only if $\mathcal{C}$ is a closed symmetric monoidal category.
\end{thm}
\begin{proof}
The proof rests on the same key point as the characterisation theorem for linked monoidal super-categories, that one can construct the curried version of any process $f$ using the fully static version $\hat{f}$ along with the basic manipulations $\otimes$, $\circ$, and $\epsilon_{I , A}$. For readability we treat $\mathcal{C}$ to be strict monoidal, so that we do not need to include static unitors in our definitions. We introduce a key process $\Delta$ named ``\textit{partial insertion}" which takes the static form $\hat{f}: I \rightarrow (C \otimes A) \Rightarrow B$ of a process $f: C \otimes A \rightarrow B$ and a state of type $C$ and then inserts that state of type $C$ into $\hat{f}$ to produce a new static process of type $\Delta(f,c):I \rightarrow (A \Rightarrow B)$.
\begin{equation}\input{figs/delta_define.tikz}\end{equation}
Using the defining equations of a Higher order process theory, $\Delta$ satisfies
\begin{equation}
  \input{figs/thetadefined_2_a.tikz}
 \end{equation}
In turn this entails that for each $f$ there exists a process $\Delta(f) := \Delta \circ (\hat{f} \otimes id)$ which satisfies
 \begin{equation}
  \input{figs/thetaproved.tikz}
 \end{equation}
 More-over by complete injectivity this $\Delta(f)$ is the unique morphism satisfying the above condition. The unique choice $\Delta(f)$ for each $f$ then satisfies the defining condition of a closed monoidal category. To show that every closed symmetric monoidal category is a HOPT all that is required is to show that $\epsilon_{I,A}$ is an isomorphism and that the sequential and parallel composition processes $\otimes$ and $\circ$ must exist. The latter is well known \cite{Johnstone1983BASIC64}, and follows by considering the left hand side of the defining equations of sequential and parallel composition processes to take the place of the arbitrary $f$ in the definition of a closed symmetric monoidal category. The two-sided inverse of $\epsilon_{I,A}$ which regards it an isomorphism is constructed by currying of the unitor of a symmetric monoidal category $\lambda: A \otimes I \rightarrow A$ to $\bar{\lambda}: A \rightarrow (I \Rightarrow A)$, in process-theoretic language, the inverse of $\epsilon_{I,A}$ is given graphically by currying the identity.
\end{proof}

\section{The existence of canonical processes of HOPTs}
In this section we prove that in any HOPT $\mathcal{C}$ morphisms satisfying the defining conditions for $d_A$, $T_{AB}$, and $\phi_{ABC}$ as defined in the main text, uniquely exist for all objects of $\mathcal{C}$.
\begin{thm}
The following hold in any HOPT $\mathcal{C}$:
\begin{itemize}
    \item For each object $A$ there exists a unique dualiser $d_A$
    \item For each pair $A,B$ there exists a unique lifting process $T_{AB}$
    \item For each triple $A,B,C$ there exists a unique static currying $\phi_{ABC}$
\end{itemize}
\end{thm}
\begin{proof}
Each proof follows by one or more applications of the existence of the curried version of \textit{any} process, guaranteed by the closed monoidal structure of $\mathcal{C}$. Since $\mathcal{C}$ is closed monoidal we know that  for every morphism $f:(A \otimes C) \rightarrow B$ there exists a unique morphism $\bar{f}:C \rightarrow (A \Rightarrow B)$ such that \begin{equation}\input{figs/evaldef.tikz} \quad = \quad \input{figs/evaldef_2.tikz}\end{equation} 
taking $f$ the right hand side of the condition we wish for $d_A$ to satisfy: \begin{equation}\input{figs/staraut_a.tikz}\end{equation} we see that $d_a$ can be taken to be the currying of the right hand side, the existence and uniqueness of such a $d_A$ are guaranteed by the defining condition of a closed monoidal category. The existence of $T_{AB}$ can be demonstrated by two applications of currying, there must exist a unique process $L$ satisfying \begin{equation}\input{figs/adjoint_exist_a.tikz}\end{equation} in turn there must be a unique process satisfying
\begin{equation}\input{figs/adjont_exist_b.tikz}\end{equation} 
together this implies there is a unique process $T$ such that \begin{equation}\input{figs/adjoint_exist_c.tikz}\end{equation}
Finally the defining condition for $\phi$: \begin{equation}\input{figs/phidefined_a.tikz}\end{equation} is again precisely the condition that $\phi$ be the currying of the morphism on the right-hand side of the condition. That such a $\phi$ exists and is unique is then again immediately implied by the closed monoidal structure of $\mathcal{C}$. 
\end{proof}
\section{Enough effects entails injective dualisation}
The following proof is a useful exercise in getting used to working with the dualiser process $d_A$. 
\begin{thm}
If an object $A$ in a HOPT has enough effects, then it has injective dualisation.
\end{thm}
\begin{proof}
let $d_A \circ \rho = d_A \circ \sigma$, we will show that $\rho$ must equal $\sigma$. This follows by using the defining properties of the insertion process and the dualising process. For every effect $e$ it follows that: \begin{equation}\input{figs/effects.tikz}\end{equation} and so by enough effects $\rho = \sigma$.
\end{proof}
\section{Proof that $\phi$ is an isomorphism}
\begin{thm}
The process \begin{equation}\phi: C \Rightarrow (A \Rightarrow B) \rightarrow (A \otimes C) \Rightarrow B\end{equation} defined by \begin{equation}\input{figs/phidefined_a.tikz}\end{equation} is an isomorphism
\end{thm}
\begin{proof}
Define the currying $\bar{\Delta}$ of $\Delta$ by 
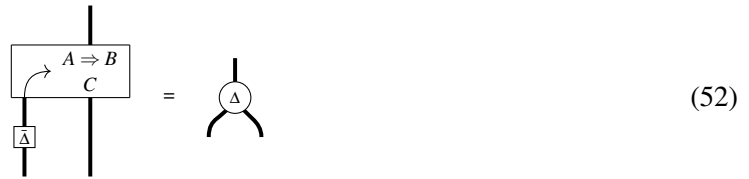
\begin{equation}\input{figs/curry_delta_a.tikz}\end{equation}
Then using complete injectivity of all insertion morphisms, its is sufficient to check that $\hat{\Delta}$ and $\phi$ are isomorphisms up to insertion. First we check that $\phi \circ \hat{\Delta} = id$
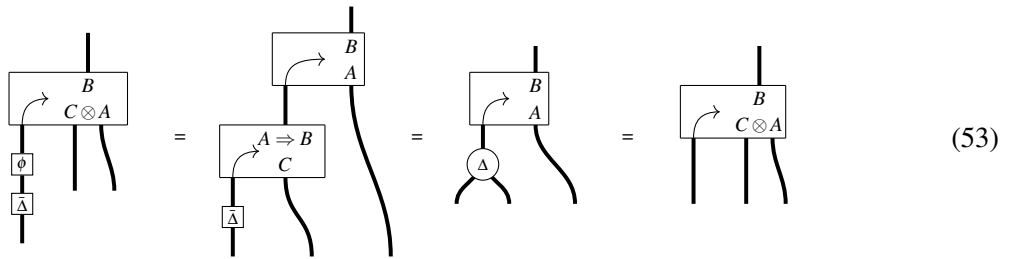
\begin{equation}\input{figs/phi_iso_1_a.tikz}\end{equation}
Then we check that $\hat{\Delta} \circ \phi = id$
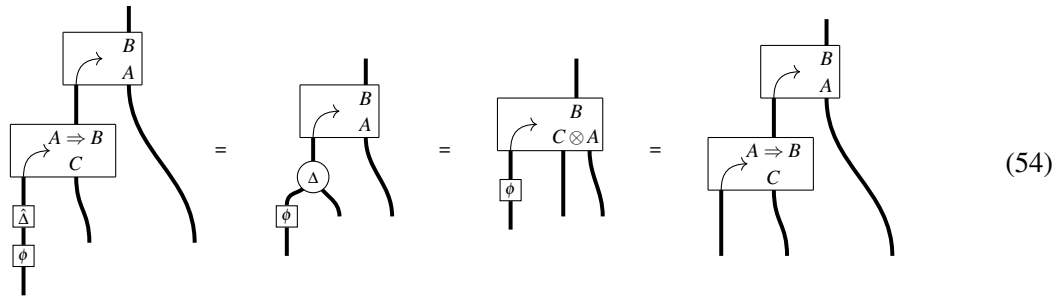
\begin{equation}\input{figs/phi_iso_2_a.tikz}\end{equation}
\end{proof}



\section{HOCCs have no correlations with single-state objects}
The notations and terminologies used here are taken from \cite{Kissinger2019AStructure}. 
\begin{thm}
Every HOCC is a HOPT which has no correlations with single-state objects
\end{thm}
\begin{proof}
A general state on $\mathbf{X} \otimes \mathbf{Y}$ is a member of the set $(C_X \times C_Y)^{**}$ where $C_X$ is the set of states on $\mathbf{X}$, $C_Y$ is the set of states on $C_Y$ and $C^{*}$ is the set of effects which normalise elements on $C$, i.e.  $\forall \rho \in c:$ $\pi \circ \rho = 1$. Let $\mathbf{Y}$ be a single-state object, since $\mathbf{Y}$ is flat its unique state must be a scalar multiple of the maximally mixed state.
\begin{equation}
  \begin{tikzpicture}[scale=0.7, {every node/.style}={scale=0.7}]
	\begin{pgfonlayer}{nodelayer}
		\node [style=none] (0) at (0.25, 0.5) {};
		\node [style=backdiscard] (1) at (0.25, -0.5) {};
		\node [style=none] (2) at (-0.5, 0) {$\alpha$};
	\end{pgfonlayer}
	\begin{pgfonlayer}{edgelayer}
		\draw (1) to (0.center);
	\end{pgfonlayer}
\end{tikzpicture}

 \end{equation}
Since $C_X$ is flat it follows that a scalar multiple of the discard process exists inside $C_X^{*}$. 
\begin{equation}
  \input{figs/tdof1.tikz}
 \end{equation}
The elements of the set $(C_X \times C_Y)^{*}$ are up to process-state duality the processes $M$ in the underlying category such that, 
\begin{equation}
  \input{figs/tdof2.tikz}
 \end{equation}
Note that any first-order causal process $\Psi$ 
\begin{equation}
  \input{figs/tdof3.tikz}
 \end{equation}
which entails that $\frac{\mu}{\alpha} \Psi \in (C_X \times C_Y)^{*}$. In turn since $\{\frac{\mu}{\alpha} \Psi ~|~ \Psi \textrm{ causal }\} \subseteq (C_X \times C_Y)^{*}$ then it follows that $(C_X \times C_Y)^{**} \subseteq \{\frac{\mu}{\alpha} \Psi~|~ \Psi \textrm{ causal }\}^{*}$. For any $w \in \{\frac{\mu}{\alpha} \Psi ~|~ \Psi \textrm{ causal }\}^{*}$ it is immediate that $\frac{\mu}{\alpha}w \in \{\Psi~|~ \Psi \textrm{ causal }\}^{*}$ which in turn implies the following decompositions,
\begin{equation}
  \input{figs/tdof4.tikz}
 \end{equation}
By assumption the usage of an effect of the form $\mathbf{Y \rightarrow I}$ (which will be normalised by the right hand side of the composition) on $w$ produces a state on $\mathbf{X}$. This in turn confirms that the left hand side of the decomposition is indeed a state of $\mathbf{X}$, and so any $w \in (C_X \times C_Y)^{**} \subseteq \{\frac{\mu}{\alpha} \Psi ~|~ \Psi \textrm{ causal }\}^{*}$ must decompose as the unique state of $\mathbf{Y}$ in parallel with a state of $\mathbf{X}$.
\end{proof}

\section{Tensor product processes vs bipartite processes}

\begin{thm}
Let $\mathcal{C}$ be a deterministic HOPT with no interaction with single-state objects. If $A$ and $B$ are causal and 
\begin{flalign}
\input{figs/tensor.tikz} : (A \Rightarrow B) \otimes (B \Rightarrow A) \rightarrow (A \otimes B) \Rightarrow (B \otimes A)
 \end{flalign}
is an isomorphism, then $A$ and $B$ are single-state objects. 
\end{thm}
\begin{proof}
We show that there exists some $\kappa': I \rightarrow A$ such that for every $\rho: I \rightarrow A$ then $\rho = \kappa'$, meaning that there can only be one state of type $I \rightarrow A$ implying that $A$ be a single-state object. Indeed for every $\rho$:
\begin{equation}
    \input{figs/trivial1_a.tikz}
 \end{equation}
it follows that every state on $A$ is equal to $\kappa'$ and so $A$ is a single-state object. Almost identical steps can be used to produce the same result for $B$.
\end{proof}

\section{Adjoint dynamics and double duals}

\begin{thm}
Let $\mathcal{C}$ be a HOPT, the following statements are equivalent. 
\begin{itemize}
    \item Every $B \in o(\mathcal{C})$ is canonically equivalent to its double dual in $\mathcal{C}$
    \item For all $A,B \in o(\mathcal{C})$ the HOPT $\mathcal{C}$ has adjoint dynamics between $A$ and $B$e 
    \item For all $B \in o(\mathcal{C})$ the HOPT $\mathcal{C}$ has adjoint dynamics between $I$ and $B$
\end{itemize}
\end{thm} 
\begin{proof}
To show that the first statement implies the second, we note that each $T_{AB}$ may be written in the following form \begin{equation}
    \input{figs/adjoint_2.tikz}
\end{equation} which is easily demonstrated by showing that the rhs indeed satisfies the defining condition for $T$ \begin{equation}
    \input{figs/adjoint_3_a.tikz}
\end{equation}
From the above decomposition of $T_{AB}$ it follows that whenever $d_B$ is an isomorphism then $T_{AB}$ is an isomorphism. The third statement immediately follows as a subcase of the second. To demonstrate that the third statement implies the first we note that each $d_B$ can be written in terms of $T_{IB}$ and a pair of isomorphisms, \begin{equation}
    \input{figs/adjoint_4.tikz}
\end{equation}
again demonstrated by showing that the rhs satisfies the defining equation for $d_B$ \begin{equation}
    \input{figs/adjoint5_a.tikz}
\end{equation}
\end{proof}

\section{Lifting isomorphism with double dual}
In this section we will use the notation $f \Rightarrow g$ to mean the supermap which pre-composes with $f$ and post-composes with $g$, the formal definition of $f \Rightarrow g$ is given in Appendix A. We will furthermore regularly use the notations $A \Rightarrow g$ and $f \Rightarrow B$ as shorthand for $id_{A} \Rightarrow g$ and $f \Rightarrow id_{B}$ respectively.
\begin{thm}[Lifted double duals]
Let $\mathcal{C}$ be any HOPT, if $\otimes$ preserves equivalence with double duals then $ \Rightarrow $ preserves equivalence with double duals.
\end{thm}
\begin{proof}
We first give a sketch proof, outlining the sequence of internal isomorphisms used to show that $((A \Rightarrow B) \Rightarrow I) \Rightarrow I \cong (A \Rightarrow B)$, we then expand on this demonstrating that the above isomorphism is actually witnessed by $d_{A\Rightarrow B}$. Firstly assuming $d_A$ and $d_B$ are isomorphisms then $d_B \Rightarrow I$ is an isomorphism since the contravariant functor $(- \Rightarrow I)$ preserves isomorphisms. Furthermore since $\otimes$ preserves equivalence with double duals  \begin{equation}(A \otimes (B \Rightarrow I)) \cong ((A \otimes (B \Rightarrow I)) \Rightarrow I) \Rightarrow I \end{equation}
Again using that $d_B$ is an isomorphism gives
\begin{flalign} 
(A \Rightarrow B) \cong (A \Rightarrow ((B \Rightarrow I) \Rightarrow I)) \cong (A \otimes (B \Rightarrow I))\Rightarrow I
 \end{flalign}
again since the contravariant functor $(- \Rightarrow I)$ preserves isomorphisms this implies,
\begin{flalign}
((A \Rightarrow B) \Rightarrow I) \Rightarrow I \cong (((A \otimes (B \Rightarrow I))\Rightarrow I) \Rightarrow I) \Rightarrow I) 
 \end{flalign}
the right hand side can be simplified using the first point on $\otimes$. 
\begin{flalign}
    ((A \Rightarrow B) \Rightarrow I) \Rightarrow I \cong (A \otimes (B \Rightarrow I))\Rightarrow I \cong A \Rightarrow B
 \end{flalign}
So there indeed exists an isomorphism of the form required, to move beyond a sketch proof it must be shown that this isomorphism is in fact $d_{A \Rightarrow B}$. Using $\phi$ and the invertible (by assumption) canonical morphism $d : B \rightarrow (B \Rightarrow I) \Rightarrow I$ in its static form $\hat{d}_{B}: I \rightarrow (B \Rightarrow ((B \Rightarrow I) \Rightarrow I))$ an invertible morphism $m$ can be built.
\begin{equation}
  \input{figs/can4.tikz}
 \end{equation}
$d_{A \Rightarrow B}$ can be expressed in terms of $m$ and $d_{(A \otimes (B \Rightarrow I)) \Rightarrow I}$ in the following way,
\begin{equation}
  \input{figs/can7.tikz}
 \end{equation}
Where since $m$ is an isomorphism $(m \Rightarrow I)$ and $(m \Rightarrow I) \Rightarrow I$ are isomorphisms too.
\begin{equation}
  \input{figs/can8.tikz}
 \end{equation}
The proofs of the identities used above can be found in Appendix A. The proof that $d_{A \Rightarrow B}$ decomposes as above is then given as follows.
\begin{equation}
  \input{figs/can5_a.tikz}
 \end{equation}
By assumptions $d_A$ and $d_B$ are isomorphisms, so $d_{B} \Rightarrow id$ is an isomorphism. It can be shown that $d_{B} \Rightarrow id$ is always the the right inverse of $d_{B \Rightarrow I}$ since first by expanding the definition of $d_{B} \Rightarrow I$ \begin{equation}
  \input{figs/can2a.tikz}
 \end{equation}
 and then using the definition of any canonical morphism $d_X$ twice.
 \begin{equation}
  \input{figs/can2b_a.tikz}
 \end{equation}
Since $d_{B} \Rightarrow id$ is an isomorphism and $d_{B} \Rightarrow id$ is a right inverse for $d_{B \Rightarrow I}$, it follows that $d_{B \Rightarrow I}$ must be an isomorphism. Since $\otimes$ preserves isomorphism with double dual $d_{A \otimes (B \Rightarrow I)}$ must be an isomorphism and by the same reasoning as for $B$ it follows that $d_{(A \otimes (B \Rightarrow I)) \Rightarrow I}$ is an isomorphism. This completes the proof that every part of the given decomposition of $d_{A \Rightarrow B}$ is then an isomorphism, entailing that $d_{A \Rightarrow B}$ itself must also be an isomorphism.
\end{proof}

\section{Wires with no-signalling states}
\begin{thm}
Let $\mathcal{C}$ be a deterministic HOPT with no correlations with single-state objects, then if
\begin{itemize}
    \item $\otimes$ preserves equivalence with double duals
    \item $A$ and $A'$ each have enough states and are canonically equivalent to their double duals
\end{itemize}
then the object $(A \Rightarrow A')$ has no-signalling states.
\end{thm}
\begin{proof}
We first show that every effect $\Pi : (A \Rightarrow A') \rightarrow I$ can be written as an application of a discard effect and an insertion of a state. This is a consequence of the isomorphism $A \otimes (A' \Rightarrow I) \cong (A \Rightarrow A') \Rightarrow I$ constructed by the following morphisms.
\begin{equation}
  \input{figs/ignoredefine.tikz}
 \end{equation}
Indeed one can show the following identity
\begin{equation}
    \input{figs/ignoridentity_a.tikz}
 \end{equation}
Using the general formula
\begin{equation}
  \input{figs/ignoremovearrow_a.tikz}
 \end{equation}
twice. 
\begin{equation}
  \input{figs/ignorederive1_a.tikz}
 \end{equation}
Then using the defining property of $d$,
\begin{equation}
  \input{figs/ignorederive2_a.tikz}
 \end{equation}
and the natural isomorphism $\phi$,
 \begin{equation}
  \input{figs/ignorederive3_a.tikz}
 \end{equation}
and the defining identity of the sequential composition supermap twice we reach 
\begin{equation}
  \input{figs/ignorederive3_b.tikz}
 \end{equation}
With this identity in mind we note that for \textit{every} effect $\Pi: (A \Rightarrow A') \rightarrow I$
\begin{equation}
  \input{figs/ignorederive4_a.tikz}
 \end{equation}
we then use the property of no correlations with single-state objects on the state highlighted on the bottom left,
\begin{equation}
  \input{figs/ignorederive4_b.tikz}
 \end{equation}
 to reach
 \begin{equation}
  \input{figs/ignorederive5_a.tikz}
 \end{equation}
This time we use no correlations with single-state objects on the bipartite state highlighted on the bottom right, 
 \begin{equation}
  \input{figs/ignorederive5_b.tikz}
 \end{equation}
this finally entails that there exists some state $f'$ such that for every effect $\Pi$.
\begin{equation}
  \input{figs/ignore_derive_6_a.tikz}
 \end{equation}
which is precisely the statement that $A \Rightarrow A'$ has no-signalling states. 

\end{proof}

\end{document}